\documentclass[floatfix, twocolumn,superscriptaddress,10pt]{revtex4-2}
\usepackage[T1]{fontenc}
\usepackage{times}
\usepackage[utf8]{inputenc}
\usepackage[resetlabels,labeled]{multibib}
\usepackage{amsmath}
\usepackage{amsthm}
\usepackage{amsfonts}
\usepackage{amssymb}
\usepackage{bm}
\usepackage{bbm}
\usepackage{bbold}
\usepackage{graphicx,epsfig,color,tcolorbox}
\usepackage{physics}
\usepackage{mathtools}
\usepackage{mathrsfs}
\usepackage{latexsym}
\usepackage[colorlinks=true,linkcolor=teal,citecolor=purple]{hyperref}
\usepackage{hyperref}
\usepackage{cleveref}
\usepackage{lipsum} 
\usepackage[ruled,vlined]{algorithm2e}
\usepackage[mathscr]{eucal}
\usepackage{eucal}
\newcommand{\ot}{\otimes}
\newcommand{\id}{\mathrm{id}}

\newcommand{\locc}{\textsc{LOCC}}
\newtheorem{theorem}{Theorem}

\newtheorem{lemma}{Lemma}

\newcommand{\chg}{\textcolor{black}}
\newcommand{\chgR}{\textcolor{black}}

\begin{document}

\title{Catalytic quantum teleportation and beyond}
\author{Patryk Lipka-Bartosik}
 \affiliation{H. H. Wills Physics Laboratory, University of Bristol, Tyndall Avenue, Bristol, BS8 1TL, United Kingdom}
\author{Paul Skrzypczyk}
\affiliation{H. H. Wills Physics Laboratory, University of Bristol, Tyndall Avenue, Bristol, BS8 1TL, United Kingdom}

\date{\today}

\begin{abstract}

% new version
\chg{In this work we address fundamental limitations of quantum teleportation -- the process of transferring quantum information using classical communication and preshared entanglement. We develop a new teleportation protocol based upon the idea of using ancillary entanglement catalytically, i.e. without depleting it. This protocol is then used to show that catalytic entanglement allows for a noiseless quantum channel to be simulated with a quality that could never be achieved using only entanglement from the shared state, \chgR{even for catalysts with a small dimension}. On the one hand, this allows for a more faithful transmission of quantum information using generic states and fixed amount of consumed entanglement. On the other hand, this shows, for the first time, that entanglement catalysis provides a genuine advantage in a generic quantum-information processing task. Finally, we show that similar ideas can be directly applied to study quantum catalysis for more general problems in quantum mechanics. As an application, we show that catalysts can activate so-called passive states, a concept that finds widespread application e.g. in quantum thermodynamics.}

% previous version
%Quantum catalysis demonstrates that in certain situations the very presence of entanglement can improve one's abilities of manipulating other entangled states. In the same time, however, it is not clear if entanglement that is used catalytically can ever provide any additional power for the existing quantum protocols. Here we show for the first time that catalysis of entanglement can provide a genuine advantage in the task of quantum teleportation. In particular, we show that in a teleportation setting extended by allowing for using catalysts, the optimal quality of teleportation is lower-bounded by a regularisation of the standard teleportation quantifier, the average fidelity of teleportation. We then show that this new regularised quantifier surpasses the standard benchmark for a variety of quantum states, meaning there exist quantum states for which entanglement manipulated in a catalytic way provides a genuine advantage over using it directly to teleport quantum states. This hints that entanglement catalysis can be a promising new avenue for exploring novel advantages in the quantum domain.

\end{abstract}

\keywords{}
\maketitle
\chg{\textbf{\textit{Introduction.---}}} Quantum entanglement leads to correlations between distant particles that cannot be explained by any classical mechanism \cite{PhysRev.47.777,bell1964einstein,Horodecki_2009}. This intricate phenomenon is nowadays seen as an indispensable resource with an enormous number of modern applications. One of the most important applications of entanglement is quantum teleportation \cite{Bennet1993}, a communication task that uses a pair of maximally entangled qubits  $\langle qq \rangle$ and two bits of communication $[c \rightarrow c]$ to simulate a noiseless quantum channel $[q\rightarrow q]$,
\begin{align}
    \label{eq:res_ineq}
    \langle qq \rangle + 2 [c \rightarrow c] \geq [q \rightarrow q].
\end{align}
The significance of the protocol is best evidenced by its widespread applicability in various areas of quantum information \cite{Pirandola2015,Ishizaka_2008,Lo_Franco_2018,Deutsch_2000}, computation \cite{Brassard_1998,Gottesman_1999,PhysRevLett.86.5188} and even general relativity \cite{PhysRevLett.106.040403,Horowitz_2004,Gottesman_2004,Lloyd_2014}. Quantum teleportation has been realised in laboratories using a variety of different technologies, including photonic qubits \cite{Bouwmeester_1997,Ursin2004,Boschi_1998,Jin2010,PhysRevLett.86.1370,PhysRevLett.92.047904}, optical modes \cite{Furusawa_1998,PhysRevA.67.032302,PhysRevA.77.022314}, nuclear magnetic resonance \cite{Nielsen_1998}, atomic ensembles \cite{Sherson_2006,Bao20347,Chen_2008}, trapped atoms \cite{Barrett2004,Olmschenk_2009,PhysRevLett.110.140403} or solid-state systems \cite{Gao_2013,Steffen_2013,Pfaff_2014}.

Realistic teleportation protocols use generic entangled states and therefore the quantum channels they simulate are inevitably noisy. In terms of the teleportation inequality (\ref{eq:res_ineq}) this means that substituting $\langle qq \rangle$ with a generic bipartite state $\langle \rho \rangle$ leads to a quantum channel that is no longer noiseless and has to be replaced with a general teleportation channel $\mathcal{N}$. A central problem of fundamental and practical significance is engineering teleportation protocols that simulate as faithfully as possible noiseless quantum channels, as measured by a natural figure of merit, the average fidelity of teleportation \cite{Popescu1994}.

The teleportation inequality (\ref{eq:res_ineq}) is perhaps the best evidence for the resource-like nature of entanglement, as it guarantees that simulating a noiseless quantum channel always consumes a pair of maximally entangled qubits. Therefore it is reasonable to expect that, in a general protocol, teleportation fidelity can only be increased at the expense of using more entanglement. Interestingly, quantum mechanics allows for a very bizarre use of entanglement, one that is already helpful without entanglement being consumed or degraded. This surprising phenomenon is called \emph{quantum catalysis} and was introduced in \cite{Jonathan_1999}, further analysed in \cite{Turgut_2007,PhysRevA.64.042314,vanDam2003,duarte2015selfcatalytic,Aubrun_2007,Aubrun_2009,PhysRevA.79.054302} and subsequently adapted to many physical settings \cite{Brand_o_2015,Ng_2015,PhysRevLett.113.150402,Bu_2016,PhysRevLett.121.190504,Sparaciari_2017,PhysRevX.7.041033,PhysRevX.8.041051,lipkabartosik2020states,wilming2020entropy,shiraishi2020quantum,kondra2021catalytic,ding2020amplifying,henao2020catalytic,takagi2021correlation,boes2019bypassing,boes2020variance}. Quantum catalysis demonstrates that access to a special entangled state (the catalyst) allows two distant parties to manipulate their entanglement in a way that would be otherwise impossible. Importantly, the catalyst is not consumed during the process, so that the parties can repeat their task again, or use it for another purpose. This makes catalysis a particularly interesting extension of the standard paradigm of local operations and classical communication (\locc). Indeed, quantum catalysis can be viewed as a paradigm shift that leads to the ultimate limits of quantum protocols under fixed resources. Since the catalyst appearing in a catalytic protocol is not depleted, it does not contribute to the overall balance of consumed resources. 

In this work we are interested in finding the ultimate such limit of quantum teleportation. More specifically, we ask what is the best teleportation fidelity that can be achieved when consuming a given entangled state and using an arbitrary amount of entanglement catalytically? We show that this natural extension of the standard teleportation protocol allows quantum channels with much larger teleportation fidelity to be achieved, or equivalently, for the transfer of quantum information much more reliably. More formally, we show that using a quantum catalyst one can achieve teleportation fidelity equal to a regularisation of the standard teleportation fidelity, which we then shown to be strictly larger than the standard teleportation fidelity for a wide range of pure states. This uncovers new limits for quantum teleportation under fixed resources. This can also be interpreted as providing the first example where quantum catalysis is successfully used in a generic information processing task, and opens up the prospects for much further investigation into the power and generality of quantum catalysis, beyond what had been appreciated up until now. In this line, we show that our methods can be adapted beyond quantum teleportation and quantum information.  

%Our proof is based on a subroutine which uses the catalyst to increase entanglement fraction of a state, which we believe to be of independent interest. In fact, we show that our methods can be used   

\textbf{\textit{Framework.---}} In what follows we % denote a discrete and finite-dimensional Hilbert space associated with a quantum system $S$ with $\mathcal{H}_{S}$. We also denote the space of all density operators acting on $\mathcal{H}_{S}$ with $\mathcal{D}(\mathcal{H}_{S})$. 
we will be interested in scenarios involving two distant parties (say Alice and Bob) who are allowed to use local operations and classical communication (\locc). We say that $\mathcal{E} \in \locc(A:B)$ if it can be written as a sequence of quantum channels applied locally by $A$ and $B$, intertwined with classical communication. To quantify entanglement we will use the entanglement fraction \cite{PhysRevA.60.1888}, which is defined as the maximal overlap with a maximally entangled state, that is
\begin{align}
    \label{ent_fra} \nonumber
    f(\rho) :=\quad \max_{\mathcal{E}}& \quad \langle \phi^+_{AB} | \mathcal{E}(\rho_{AB})  |\phi^+_{AB}\rangle \\
    \text{s.t.} &\quad \mathcal{E} \in \locc(A:B),
\end{align}
where $\ket{\phi^+_{AB}} = \sum_{i=1}^{d} \ket{i}_A\!\ket{i}_B / \sqrt{d}$ denotes a maximally-entangled state shared between $A$ and $B$. 

{\textit{Standard quantum teleportation.---}} Before presenting our main results let us briefly recall the task of quantum teleportation \cite{Bennet1993}. In its most general form, the protocol involves two spatially separated parties, Alice and Bob, who share an arbitrary quantum state $\rho_{{AB}}$ of dimension $d_{{A}} \times d_{{B}}$. A third party, often called a referee, provides Alice with a quantum state $\varphi_R$ of dimension $d_{{R}}$ which is unknown to both parties. The goal set before Alice and Bob is to transfer the unknown state from one party to another, using only local operations and classical communication, i.e. quantum channels $\mathcal{T} \in \locc(RA:B)$,  and shared entanglement. Under this conditions all possible states which can be achieved in Bob's lab can be written as:      
\begin{align}
    \rho'_B = \tr_{RA}\mathcal{T}(\varphi_R \ot \rho_{AB}).
\end{align}
The above protocol can be viewed equivalently as a process of establishing a quantum channel between Alice and Bob that maps the input state $\varphi_R$ to the output $\rho_{B}'$. \chg{The goal of quantum teleportation is then to simulate a noiseless quantum channel between Alice and Bob, i.e. an identity map $\id_{A \rightarrow B}$}. The quality of teleportation or, equivalently, the fidelity of the resulting teleportation channel, can be quantified using the average fidelity of teleportation \cite{Popescu1994} (or simply ``fidelity of teleportation'') defined as
\begin{align}
    \label{eq:fid_tel}
    \langle F \rangle_{\rho} := \max_{\mathcal{T}} & \, \int \langle \varphi|\tr_{RA}\mathcal{T}(\varphi_R \ot \rho_{AB})|\varphi\rangle \, \text{d} \chg{\varphi} \nonumber \\
    \text{s.t.}& \quad \mathcal{T}\in \locc(RA:B).
\end{align}
The integral in (\ref{eq:fid_tel}) is computed over a uniform distribution of all pure input states $\varphi = \dyad{\varphi}$ according to a normalised Haar measure $\int \text{d} \varphi = \mathbb{1}$.   It can be easily verified that $0 \leq \langle F \rangle_{\rho} \leq 1$ for all density operators $\rho$. Furthermore, the case $\langle F\rangle_{\rho} = 1$ corresponds to perfect teleportation from Eq. (\ref{eq:res_ineq}) and is possible if and only if $\rho$ is maximally-entangled. In practice, the fidelity of teleportation will always be less than one. Furthermore, when the shared state is separable, the corresponding teleportation protocol is said to be “classical” and fidelity of teleportation is bounded by $\langle F \rangle_c := 2/( d_{R} + 1)$. Importantly, it was shown in Ref. \cite{PhysRevA.60.1888} that fidelity of teleportation (\ref{eq:fid_tel}) is related with entanglement fraction (\ref{ent_fra}) via:
\begin{align}
    \langle F \rangle_{\rho} = \frac{f(\rho) d_R+1}{d_R+1}.
\end{align}
In what follows we will focus on this quantity and show that catalysts allow to increase entanglement fraction without consuming any additional entanglement.

\textbf{\textit{Results.---}} Let us start by describing a catalytic extension of the general quantum teleportation protocol. Then we display a main theorem that gives a lower bound on its performance and show that the bound is tight enough to demonstrate a sharp advantage with respect to the standard teleportation protocol. \chg{We conclude with a simple generalization of these methods that can be used to address catalytic advantages in more general settings.} 

\textit{Catalytic quantum teleportation.---} Assume that Alice and Bob, in addition to their shared state $\rho_{AB}$, have access to a quantum system $CC'$ prepared in a state  $\omega_{CC'}$. This additional system is distributed such that Alice has access only to $C$, and Bob only to its $C'$ part. Alice is then given an unknown quantum state $\varphi_R$ and the parties perform a protocol $\mathcal{T} \in \locc(RAC:BC')$ which now acts on both systems they share and the input system. Moreover, for the protocol to be catalytic we demand that $\mathcal{T}$ does not modify the catalyst. Notably, we do allow the catalyst to become correlated with $\rho_{AB}$ during the process and in the Appendix we show that these correlations can be made arbitrarily small in trace distance, at the expense of using larger catalysts. The final state of Bob's subsystem at the end of the catalytic teleportation protocol reads:
\begin{align}
    \rho'_{B} = \tr_{RACC'}\left[\mathcal{T}(\varphi_R \ot \rho_{AB} \ot \omega_{CC'})\right]
\end{align}
The quality of the protocol can be quantified similarly as in the case of standard teleportation, i.e. using the fidelity of teleportation (\ref{eq:fid_tel}). Since we have the freedom to choose the catalyst, we define the fidelity of catalytic teleportation $\langle F_{\text{cat}} \rangle_{\rho}$ as
\begin{align}
    \nonumber
    \langle F_{\text{cat}} \rangle_{\rho} = \max_{\mathcal{T}\!,\, \, \omega}& \int \langle \varphi|\tr_{RACC'}\mathcal{T}(\varphi_R \!\ot \!\rho_{AB} \ot \omega_{CC'})|\varphi\rangle \, \text{d} \varphi \! \\
    \text{s.t.}& \quad \tr_{RAB}\mathcal{T}(\varphi_R \ot \rho_{AB} \ot \omega_{CC'}) = \omega_{CC'}, \nonumber\\
    &\quad \mathcal{T} \in \locc(RAC:BC'),  \nonumber \\  
    &\quad \omega_{CC'} \geq 0, \quad \tr \omega_{CC'} = 1. \label{eq:fid_cat_tel}
\end{align}
Let us now define a regularisation of the entanglement fraction from (\ref{ent_fra}), which we will denote by $f_{\text{reg}}(\rho)$, and whose significance will soon become evident, namely
\begin{align}
  \label{eq:reg_ent_frac}
  f_{\text{reg}}(\rho) := & \lim_{n \rightarrow \infty} \frac{f_n(\rho^{\ot n})}{n},
\end{align}
where $f_n(\sigma)$ is the solution to
\begin{align}
    f_{n}(\sigma) :=  \max_{\mathcal{E}}& \,\,\, \sum_{i=1}^n \langle\phi^+ | \tr_{/i}\mathcal{E}(\sigma)|\phi^+\rangle, \nonumber \\
    \text{s.t.}&\,\,\, \mathcal{{E}} \in \locc(A_1 \ldots A_n:B_1\ldots B_n),
    \label{eq:fn_ent_frac}
\end{align}
where $\tr_{/i}(\cdot)$ is the partial trace performed over particles $1\ldots i-1, i+1\ldots n$. Notice that by taking a sub-optimal guess $\mathcal{E} = \mathcal{E}_1 \ot \mathcal{E}_2 \ot \ldots \ot \mathcal{E}_n$ with $\mathcal{E}_1 = \mathcal{E}_2 = \ldots = \mathcal{E}_n$ we can infer that $f_{\text{reg}}(\rho) \geq f(\rho)$ for all quantum states $\rho$. With the above definitions we are now ready to present our main result.
\begin{theorem} The fidelity of catalytic teleportation satisfies
\label{thm1}
\begin{align}
    \label{eq:f_cat}
    \langle F_{\emph{cat}} \rangle_{\rho}  \geq  \frac{f_{\emph{reg}}(\rho) d_R + 1}{d_R + 1}
\end{align}
In other words, there is a protocol $\mathcal{T} \in \locc(RAC:BC')$ and a catalyst $\omega_{CC'} $ that achieves the bound (\ref{eq:f_cat}).
\end{theorem}

\begin{proof}
We will sketch the proof of Theorem \ref{thm1} (see Appendix for formal derivation). We start by constructing the catalyst and a subroutine $\mathcal{T}_{\mathcal{E}}$ that increases entanglement fraction of $\rho_{AB}$. Then we use this preprocessed state to perform standard teleportation $\mathcal{T}'$. The total protocol then reads $\mathcal{T} = \mathcal{T}' \circ \mathcal{T}_{\mathcal{E}}$.

Let $n \geq 2$ be a finite natural number and denote $C := C_2\ldots C_nM$ and $C' := C'_2 \ldots C'_nM$, where $M$ is a classical register. Moreover, let $\mathcal{E} \in \locc(AC:BC')$ be a channel (yet to be determined) and denote $\sigma^{n-i} := \tr_{1\ldots i}\mathcal{E}(\rho^{\ot n})$, where $\tr_{1\ldots i}(\cdot)$ denotes partial trace over the first $i$ copies of $\rho^{\ot n}$. Consider the following catalyst, introduced in \cite{Duan_2005}: 
\begin{align}
    \label{eq:duan_state}
    \omega_{CC'} = \frac{1}{n}\sum_{i=1}^{n} \underbrace{\rho^{\ot i} \ot \sigma^{n-i}}_{C_2C_2' \ldots C_nC_n'} \ot \dyad{i}_M,
\end{align}
\chg{This catalyst is a state of $n-1$ quantum registers, each of dimension $d$, and a classical register of dimension $n$. For a given value $i$ of the classical register, the remaining $n-1$ quantum registers contain $i$ copies of the shared bipartite state $\rho$, and an $n-i$-partite state $\sigma^{n-i}$ that is the marginal of $\mathcal{E}(\rho^{\ot n})$.}

Let us label for clarity $A_1 \equiv A$ and $A_i \equiv C_i$ for $2\leq i\leq n$ and similarly for $B_i$ and $BC_2\ldots C_n$. The joint state of the resource and the catalyst, $\rho_{AB} \ot \omega_{CC'}$, is presented in Fig. \ref{fig:1}a for the exemplary case with $n = 5$. The initial protocol $\mathcal{T}_{\mathcal{E}}$ can be summarised as follows.
\begin{enumerate}
    \item Apply $\mathcal{E} \in \locc(AC:BC')$ to the $n$-th pair using $M$ as the control (see Fig. \ref{fig:1}b).
    \item Relabel $\ket{i}_M \rightarrow \ket{i+1}_M$ for $i < n$ and $\ket{n}_M \rightarrow \ket{1}_M$ (see Fig. \ref{fig:1}c).
    \item Relabel $A_1B_1 \rightarrow A_iB_i$ for all $i$ in $M$ (see Fig. \ref{fig:1}d).
    \item Discard the catalyst $CC'$
\end{enumerate}
As a result the system and the catalyst transform into
\begin{align}
    \nonumber
    \rho_{AB} \rightarrow \rho_{AB}^{(n)} = & \tr_{CC'}\mathcal{T}_{\mathcal{E}}(\rho_{AB} \ot \omega_{CC'}) \\ =& \frac{1}{n} \sum_{i=1}^n \tr_{/ i} \mathcal{E}(\rho^{\ot n}_{AB}),\\
    \omega_{CC'} \rightarrow \omega_{CC'}' =& \tr_{AB} \mathcal{T}_{\mathcal{E}}(\rho_{AB} \ot \omega_{CC'}) = \omega_{CC'}.
    \label{eq:subroutine}
\end{align}

\begin{figure}
    \centering
    \includegraphics[width=\linewidth]{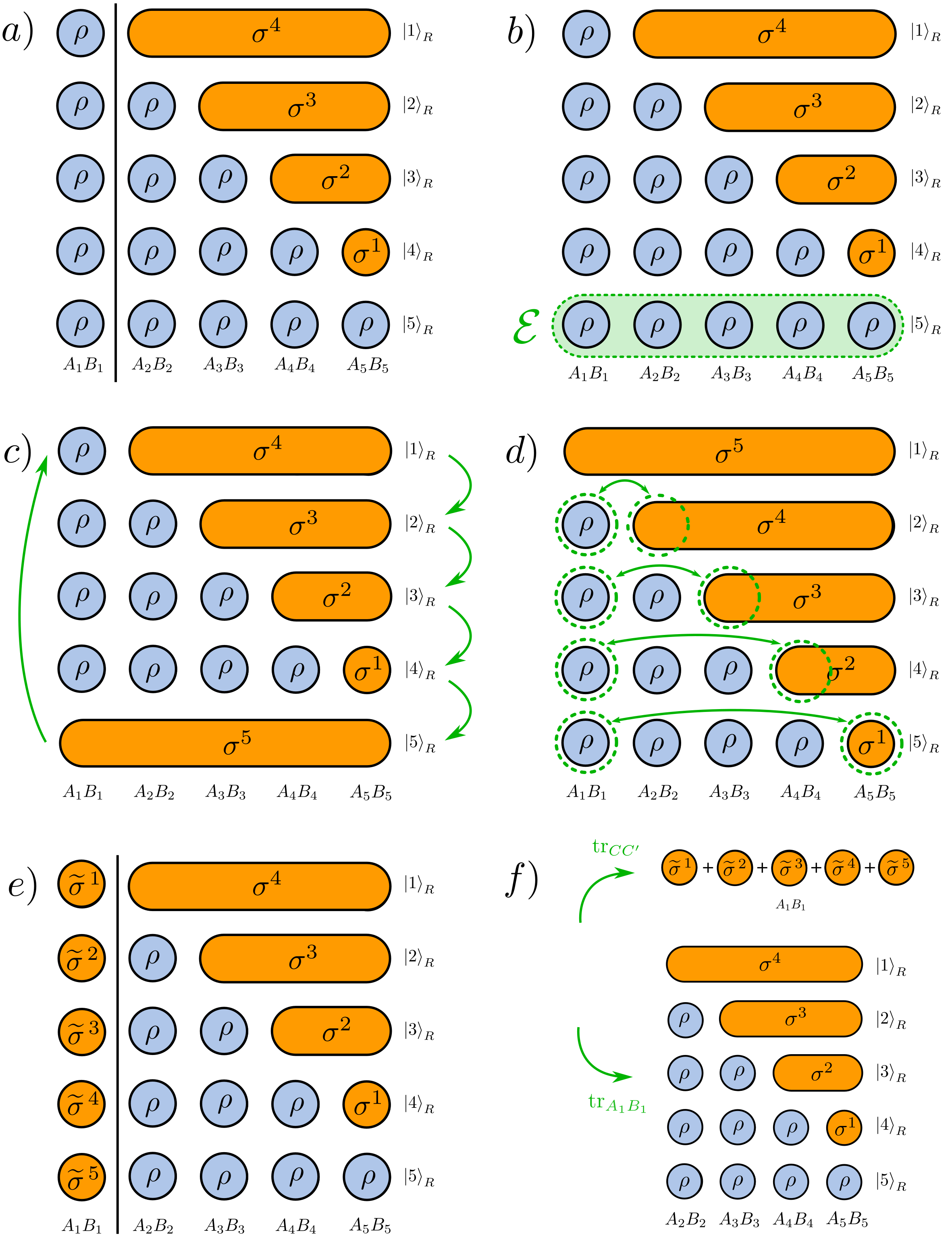}
    \caption{The catalytic subroutine that uses a noisy entangled state as a catalyst to enhance entanglement fraction. Panels $(a)-(e)$ describe subsequent steps of the protocol and $(f)$ the final state of the main system and the catalyst. The catalyst remains unchanged as the system is transformed into a state with a higher entanglement fraction.}
    \label{fig:1}
\end{figure}

Next we apply the standard teleportation scheme for noisy states \cite{Horodecki1998}. Taking $\{U_a^{ A}\}$ for $a \in \{1, \ldots, d_A^2\}$ with $d_A = d_B = d_R$ to be the set of generalised Pauli operators with respect to basis $\{\ket{i}^{ A}\}$ we can summarise $\mathcal{T}'$ as follows. 
\begin{enumerate}
    \item Twirl $\rho_{AB}$ into an isotropic state: 
    \begin{align}
        \rho_{AB}^{(n)} \rightarrow  f(\rho_{AB}^{(n)}) \phi^+_{AB} + [1 - f(\rho_{AB}^{(n)})]\phi_{AB}^{\bot},
    \end{align}
    where $\phi^{\bot} = (\mathbb{1} - \phi^+)/(d_A^2-1)$.
    \item Perform teleportation on $RA \rightarrow B$:
    \begin{enumerate} 
        \item Alice measures $RA$ using a POVM with elements:
        \begin{align}
            M_a^{RA} = (\mathbb{1} \ot U_{a}) \phi^+_{RA} (\mathbb{1} \ot U_{a}^{\dagger}),
        \end{align}
        \item Alice communicates outcome $a$ to Bob,
        \item Bob applies $U_a^{\dagger}(\cdot)U_{a}$ to his share of the state.
    \end{enumerate}
\end{enumerate}
The fidelity of teleportation in this process reads:
\begin{align}
    \frac{f(\rho_{AB}^{(n)}) d_R+1}{d_R+1}.
\end{align}
Notice that so far the channel $\mathcal{E}$ was arbitrary, and so we can now optimize $\mathcal{T} = \mathcal{T}' \circ \mathcal{T}_{\mathcal{E}}$ over all feasible channels $\mathcal{E} \in \locc(AC:BC')$. Taking the limit $n \rightarrow \infty$ and using $\lim_{n \rightarrow \infty} f(\rho_{AB}^{(n)}) = f_{\text{reg}}(\rho_{AB})$ leads to Eq. (\ref{eq:f_cat}). 
\end{proof}

\chg{The regularised entanglement fraction appears difficult to compute in general. However, for large $n$, one can use typicality arguments to find a wide range of states for which the lower-bound in Eq. (\ref{eq:f_cat}) still demonstrates a significant advantage over standard teleportation.}

%Finally, let us mention how our result modifies the general (single-shot) resource inequality for teleportation, i.e.
%\begin{align}
%    \langle \rho \rangle + \log d [c \rightarrow c] + \geq [\mathcal{N}],
%\end{align}
%where $[\mathcal{N}]$ is a quantum channel with entanglement fidelity $F(\mathcal{N})$ given by Eq. (). On the other hand, in the present work we show the following resource inequality:
%\begin{align}
%    \langle \rho \rangle + C [c \rightarrow c] + \langle cc \rangle \geq [\mathcal{N}] + \langle cc \rangle. 
%\end{align}

% Before we describe the catalytic setting let us briefly clarify in what way our framework differs from the one originally considered in \cite{Jonathan_1999}. First, we demand that at the end of the protocol the catalyst is not changed, though we allow it to be correlated with the shared state. Furthermore, we allow both the catalyst and the shared state to be a noisy entangled state, rather than a pure state. 

\textit{Demonstrating catalytic advantage in teleportation.---} Our reasoning so far was valid for arbitrary bipartite density operators $\rho_{AB}$. In this section we will restrict our attention to pure states $\rho_{AB} = \dyad{\psi_{AB}}$ and use typicality arguments to infer that the
% and use the fact that when given access to sufficiently many i.i.d copies of a state it is possible to transform it into any other state with a smaller entanglement entropy. It turns out that there are states which, although they have a smaller entanglement entropy, they perform better in the teleportation task. 
presented protocol for catalytic teleportation leads to a generic advantage over the standard teleportation protocol. Interestingly, this is a consequence of an essential property of catalysis: that certain catalysts amplify typical properties of states, even at the level of a single copy. This property of catalysts has been recently employed in \cite{lipkabartosik2020states,shiraishi2020quantum,wilming2020entropy,kondra2021catalytic,takagi2021correlation}. 
\begin{lemma}
\label{lem:1}
The regularised entanglement fraction $f_{\text{reg}}(\psi_{AB})$ for pure states $\psi_{AB}$ satisfies
\begin{align}
    \label{eq:lem1}
    f_{\text{reg}}(\psi_{AB})\geq\,\,\, \max_{\psi'}& \,\,\, f(\psi_{AB}') \\
    \emph{s.t.}& \quad S(\rho_A) \geq S(\rho_A'),
\end{align}
where $\rho_{A} = \tr_B \psi_{AB}$ and $\rho_{A}' = \tr_B \psi_{AB}'$ and $S(\rho) = - \tr \rho \log \rho$ is the Shannon entropy.
\end{lemma}
% The state $\psi_{AB}'$ from the above lemma can be interpreted as produced in the first part of the catalytic teleportation protocol $\mathcal{T}_{\mathcal{E}}$, that is $\psi_{AB}' = \tr_{CC'} \mathcal{T}_{\mathcal{E}}(\psi_{AB} \ot \omega_{CC'})$, where $\omega_{CC'}$ is the Duan state (\ref{eq:duan_state}) for a sufficiently large $n$. 
We now apply Lemma $1$ to show that catalytic teleportation outperforms standard teleportation for a wide range of generic quantum states.

\textit{{Example.---}} As a simple example let us consider teleporting a three-dimensional quantum system ($d_{R} = 3$) using a singlet. In this case the state shared between Alice and Bob can be written as $\psi_{AB} = \sum_{i=1}^{3} \sqrt{\lambda_i} \ket{i}_A\ket{i}_B$, with Schmidt coefficients $\lambda_1 = 1/2$, $\lambda_2 = 1/2$ and $\lambda_3 = 0$. Its entanglement fraction is equal to $f(\psi_{AB}) = (\sum_{i=1}^{3} \sqrt{\lambda_i} )^2/3 = 2/3$ and therefore its fidelity of teleportation reads
\begin{align}
    \langle F \rangle_{\psi} = 0.75,
\end{align}
which is also larger than the classical threshold $\langle F_c \rangle = 1/2$. 

Let us now analyse the protocol for catalytic teleportation. In this case the relevant benchmark is the fidelity of catalytic teleportation (\ref{eq:fid_cat_tel}) whose lower bound can be found using Lemma \ref{lem:1}. To compute it, we choose the optimizer in (\ref{eq:lem1}) to be the state $\psi_{AB}'$ with Schmidt coefficients $\lambda_1' = x$ and $\lambda_2' = \lambda_3' = (1-x)/2$, where $x$ is the unique solution to $h(x) = x \log 2$ (which is $x \approx 0.77$) and $h(x) = - x \log x - (1-x)\log(1-x)$. This is a feasible choice since the entropies of marginals of $\psi_{AB}$ and $\psi_{AB}'$ are both equal to $\log 2$. According to Lemma \ref{lem:1}, the regularised entanglement fraction can be lower-bounded by the entanglement fidelity of $\psi_{AB}'$, therefore $f_{\text{reg}} \geq f(\psi_{AB}') \approx 4/5$. Using Theorem \ref{thm1} we can then infer that 
\begin{align}
    \langle F_{\text{cat}} \rangle \geq 0.85,
\end{align}
which is roughly $13\%$ larger than the best fidelity that could ever be obtained when using $\psi_{AB}$ alone. Interestingly, this simple example is not a singualar case: there are in fact many entangled states whose performance in teleportation can be catalytically enhanced. To show this in Fig. \ref{fig:2} we used Lemma \ref{lem:1} and numerically computed the lower-bound on the catalytic advantage $\eta(\psi) := (\langle F_{\text{cat}} \rangle - \langle F \rangle) / \langle F \rangle$. \chgR{In the Appendix we further show that a similar advantage is present when using a small catalyst (qutrit). In that case the enhancement is around 2.5\% of what can be achieved using $\psi_{AB}$ only.}

\begin{figure}
    \centering
    \includegraphics[width=\linewidth]{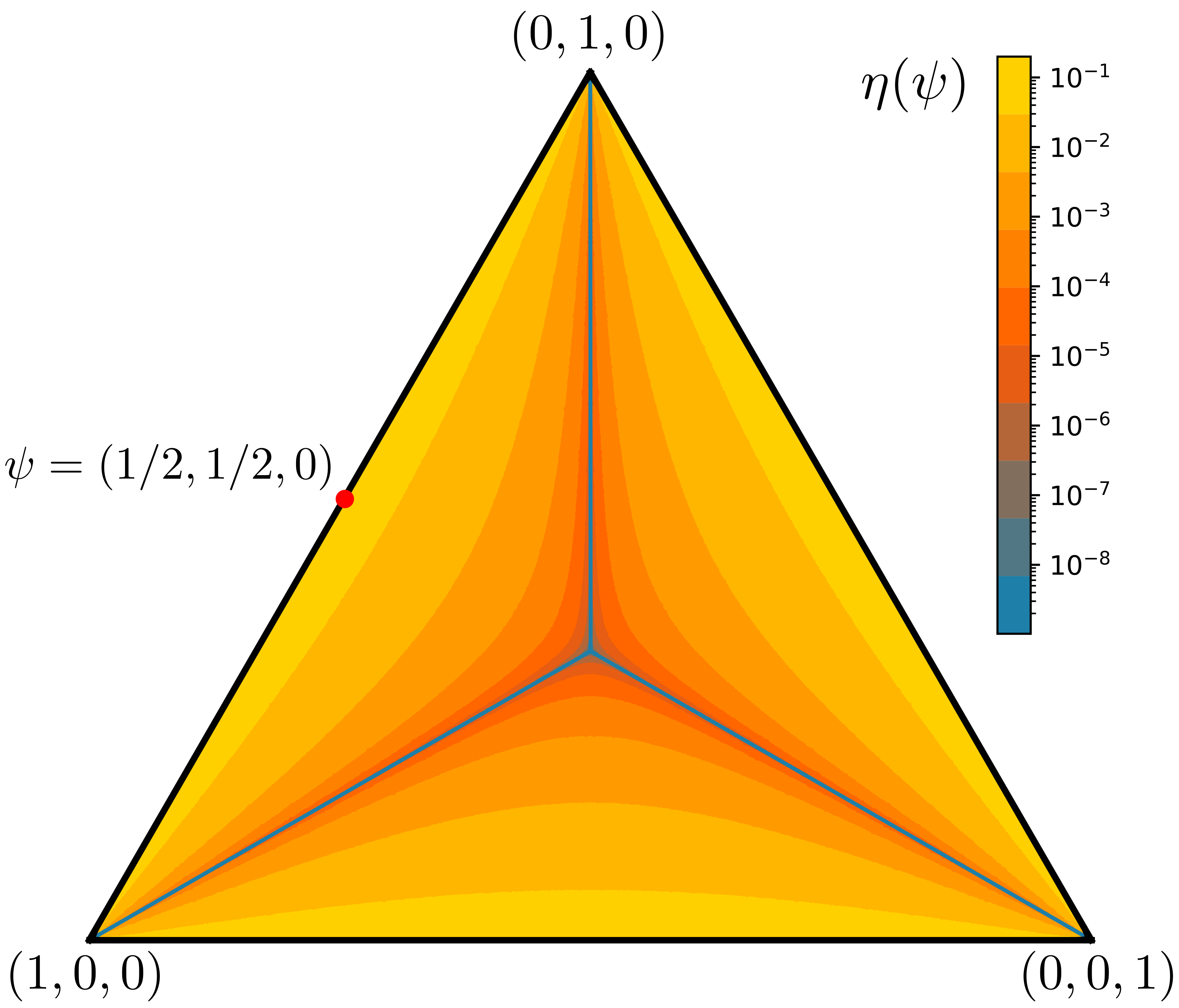}
    \caption{The catalytic advantage $\eta(\psi)$ in quantum teleportation. The triangle describes the space of all pure states of two qutrits. In particular, each point $\bm{\lambda} = (\lambda_1, \lambda_2, \lambda_3)$  corresponds to a unique (up to local unitaries) state with Schmidt coefficients $\{\lambda_i\}$ for $1 \leq i \leq 3$. The red point corresponds to the example from the main text.}
    \label{fig:2}
\end{figure}

\chg{
{\textit{Beyond quantum teleportation.---}}} The catalytic subroutine $\mathcal{T}_{\mathcal{E}}$ we used to prove Theorem \ref{thm1} can be used to address more general problems, beyond increasing the entanglement fraction, in various paradigms -- other than LOCC. Let us mention the general idea, postponing the details and an explicit application to the Appendix. 

Let us for simplicity focus on the case of a single party $S$ and let $\mathcal{E}$ be any channel from a class of channels $\mathscr{C} \subseteq \text{CPTP}$ acting on $S$. Moreover, let $O$ be an arbitrary observable on $S$. Our goal is to minimize (or maximize) the expectation of $O$ in the state $\rho$ under the available class of operations $\mathscr{C}$. Define
\begin{align}
    \label{eq:gen_erg}
    \mathcal{R}(\rho) :=  \min_{\mathcal{E} \in \mathscr{C}}\,\, \tr\left[\mathcal{E}(\rho) O\right].
\end{align}
Let us also define an analogous quantity for when many copies of $\rho$ are processed collectively, i.e.
\begin{align}
    \label{eq:gen_erg_col}
    \mathcal{R}_{\text{col}}(\rho) := \min_{\mathcal{E} \in \mathscr{C}}\,\, \frac{1}{n} \tr[\mathcal{E}(\rho^{\otimes n}) O^{\otimes n}],
\end{align}
where now $\mathcal{E}$ is a collective operation that acts on $n$ copies of $\rho$, and $O^{\otimes n} \equiv \sum_{i=1}^n \mathbb{1}_{/i} \otimes O_{i}$. Very often $\mathcal{R}_{\text{col}}(\rho) < \mathcal{R}(\rho)$, i.e. processing multiple copies collectively is strictly better than processing them one by one. Interestingly, the same improvement in manipulation abilities can be achieved when using only a single copy of $\rho$ and a suitably chosen catalyst. This is the content of our next theorem. 
\begin{theorem}
\label{theorem2}
Let $\rho$ be a quantum state and $\mathcal{D} \in \mathscr{C}$. Then there is a quantum state $\omega$ such that
\begin{align}
    \min_{\mathcal{D} \in \mathscr{C}} \tr\left[\mathcal{D}(\rho \ot \omega) (O \ot \mathbb{1})\right] = \mathcal{R}_{\emph{col}}(\rho),
\end{align}
and moreover
\begin{align}
    \tr_{1}\left[ \mathcal{D}(\rho \ot \omega)\right] = \omega.
\end{align}
\end{theorem}
Notice that by taking $S = AB$, $\mathscr{C} = \locc(A:B)$ and $O = \phi^+_{AB}$ we obtain the catalytic subroutine from the previous section. Interestingly, the reasoning presented above is much more general, and to demonstrate this in the Appendix we apply Theorem \ref{theorem2} to the problem of work extraction in quantum thermodynamics. As a consequence, it can be shown that catalysis unlocks the energy contained in a passive state under arbitrary classes of operations, therefore generalising the main result from \cite{Sparaciari_2017}.

\textbf{\textit{Discussion.---}} We have introduced an extension of the standard teleportation protocol, to the case when Alice and Bob use entangled states in a catalytic way. We showed that when arbitrary catalysts are allowed, the teleportation fidelity can be lower-bounded by a regularisation of the standard teleportation fidelity. We then showed that this regularised quantifier is strictly larger than the standard teleportation fidelity, therefore demonstrating a genuine catalytic advantage for a wide range of quantum states.

Quantum teleportation is one of many information-theoretic protocols whose performance depends directly on the entanglement fraction of the used resource. Our new methods (in particular the catalytic subroutine) can be therefore directly applied to study other protocols whose performance is quantified using entanglement fraction (see e.g. \cite{Ekert1991,Zukowski1993,Bennett1996,Bennett1992}). 

\chg{The generalised version of our catalytic subroutine, in a certain sense, allows collective effects to be incorporated at a single-copy level, using appropriately chosen catalysts. Since quantum advantages generally result from the ability of processing many quantum states simultaneously, we hope that the methods described here will lead to interesting extensions of quantum protocols that enjoy the performance of collective processing, but using only few copies of the resource.}

\chg{Finally, we believe that catalysis can lead to interesting extensions of standard quantum resource theories. Since entanglement fraction can be viewed as one of the Renyi entropies, it is plausible to expect that correlated catalysis can be used to selectively increase other Renyi entropies. This can potentially lead to better performances in various operational tasks, ranging from standard discrimination \cite{Napoli2016,Piani2016,Takagi2019,Piani2015,DS} up to more exotic variants thereof \cite{ducuara2021quantum}.}

\textit{\textbf{Acknowledgments}.---}
We thank Tulja Varun Kondra, Chandan Datta and Alex Streltsov for insightful comments and discussions on the first draft of this manuscript. We are especially grateful for pointing out a mistake in our proof of the statement that the system-catalyst correlations vanish in the limit of large catalysts (Appendix B), as well as suggesting a fruitful way to amend this problem using trace distance (see also \cite{kondra2021catalytic} for an independent proof of an analogous statement).

PLB acknowledges support from the UK EPSRC (grant no. EP/R00644X/1). PS acknowledges support from a Royal Society URF (UHQT).

%The main difference between the above setting and the setting of entanglement activation \cite{Masanes2006} is that the catalyst state $\omega_{CC'}$ cannot be modified by any procedure that happens in Alice's and Bob's labs, though in the process the catalyst can get correlated with the state $\rho_{AB}$ they share or the input state $\psi^V$.  

% can we use entanglement to activate teleportation?

\bibliographystyle{apsrev4-2}
\bibliography{citations}

\appendix
\onecolumngrid
\section{Proof of Theorem 1}
Consider two spatially separated parties $A=A_1\ldots A_n$ and $B=B_1 \ldots B_n$ composed of $n$ identical and independently distributed (i.i.d) entangled states. Let $\rho_{A_1B_1}$ be a single-copy state of the system shared by Alice and Bob and denote $A = A_1 A_2 \ldots A_n$ and $B = B_1 B_2 \ldots B_n$. We will also use the shorthand ${1\!:\!i}$ to denote a subset of subsystems $A_1 B_1\ldots A_iB_i$, with a convention that $1\!:\!0$ is an empty subset. Let us consider an arbitrary map $\mathcal{E} \in \locc(A:B)$ and denote:
\begin{align}
    \mathcal{E}(\rho^{\ot n}_{AB}) = \sigma_{AB}^n \qquad \text{and} \qquad \sigma^{n-i} := \tr_{{1:n-m}}  \left(\sigma_{AB}^n\right).
\end{align}
Consider the following state of the catalyst:
\begin{align}
    \omega_{CC'} := \sum_{i = 1}^{n} \frac{1}{n} \underbrace{\rho^{\ot i-1} \ot \sigma^{n-i}}_{A_2B_2 \ldots A_nB_n} \ot \dyad{i}_{M}.
\end{align}
In what follows we will think about the system in register $A_1B_1$ as the main state shared between Alice and Bob, and $CC'$ with $C = A_2 \ldots A_n M$ and $C' = B_2 \ldots B_n M$ as the catalyst they share, with $M$ being a classical register. The joint initial state shared between Alice and Bob is of the form:
\begin{align}
    \rho_{A_1B_1} \ot \omega_{\rm CC'} &= \frac{1}{n}\left(\rho_{A_1B_1}\ \ot \sigma^{n-1}_{2:n} \ot \dyad{1}_{M} + \ldots + \rho^{\ot n-1}_{1:n-1} \ot \sigma^{1}_{A_nB_n} \ot \dyad{n-1}_{M} + \rho^{\ot n}_{1:n} \ot \dyad{n}_{M}\right) \\
    &= \sum_{i = 1}^{n} \frac{1}{n} \underbrace{\rho^{\ot i} \ot \sigma^{n-i}}_{A_1B_1 \ldots A_nB_n} \ot \dyad{i}_{M}
\end{align}
The protocol $\mathcal{T}_{\mathcal{E}}$ that improves the fully entangled fraction of a state $\rho_{A_1 B_1}$ can be summarised as follows:
\begin{enumerate}
    \item Alice and Bob apply $\mathcal{E} \in \locc(A:B)$ conditioned on the value in the classical register $M$. More specifically, the map is of the form: 
    \begin{align}
        \id_{AB}(\cdot) \ot \sum_{i=1}^{n-1} \langle i |\cdot |i \rangle_M + \mathcal{E}_{AB}(\cdot) \ot \langle n | \cdot | n \rangle_M.
    \end{align}
    \item Alice and Bob relabel their shared classical register $M$ in the following way:
    \begin{align}
        \dyad{i}_M &\rightarrow \dyad{i+1}_M \qquad \text{for} \quad i < n, \\
        \dyad{n}_M &\rightarrow \dyad{1}_M.
    \end{align}
    \item Alice and Bob relabel their quantum systems conditioned on the classical register in the following way:
    \begin{align}
        \rho_{A_1B_1 A_2 B_2 \ldots A_iB_i A_{i+1}B_{i+1} \ldots A_n B_n} \ot \dyad{i}_M \rightarrow  \rho_{A_{i+1}B_{i+1} A_2 B_2 \ldots A_iB_i A_{1}B_{1} \ldots A_n B_n} \ot \dyad{i}_M \qquad \text{for} \quad 0 \leq i \leq n
    \end{align}
\end{enumerate}
The state shared between Alice and Bob during each of the steps of the protocol $\mathcal{T}_{\mathcal{E}}$ can be written as:
\begin{align}
    \rho_{A_1B_1} \ot \omega_{\rm CC'} \xrightarrow[]{\quad 1\quad }& \,\,
    %%%%%%%%%%%%%%%%%%%%%%%%%%%%%%%%%%%%
    % \frac{1}{n}\left(\rho_{A_1B_1} \ot \sigma^{n-1}_{2:n} \ot \dyad{1}_{R} + \ldots + \mathcal{E}(\rho^{\ot n}_{1:n}) \ot \dyad{n}_{R}\right) \\
    \sum_{i=1}^{n-1} \frac{1}{n} \rho^{\ot i}_{1:i} \ot \sigma^{n-i}_{i+1:n} \ot \dyad{i}_M + \frac{1}{n}  \mathcal{E}(\rho^{\ot n}_{1:n}) \ot \dyad{n}_{M} \\
    %%%%%%%%%%%%%%%%%%%%%%%%%%%%%%%%%%%
    =&\,\, \frac{1}{n}\left(\rho \ot \sigma^{n-1} \ot \dyad{1}_{M} + \ldots + \sigma^{n} \ot \dyad{n}_{M}\right) \\
    %%%%%%%%%%%%%%%%%%%%%%%%%%%%%%%%%%%
    \xrightarrow[]{\quad 2\quad }& \,\,  \sum_{i=0}^{n-1} \frac{1}{n} \rho^{\ot i}_{1:i} \ot \sigma^{n-i}_{i+1:n} \ot \dyad{i+1}_M
    \\
    %%%%%%%%%%%%%%%%%%%%%%%%%%%%%%%%%%%
    =&\,\, \frac{1}{n}\left(\sigma^{n} \ot \dyad{1}_{M} + \ldots + \rho^{\ot (n-1)} \ot \sigma^{1} \ot \dyad{n}_{M}  \right)  \\
    %%%%%%%%%%%%%%%%%%%%%%%%%%%%%%%%%%%
    \xrightarrow[]{\quad 3\quad }& \,\,  \frac{1}{n} \sigma^n_{1:n} \ot \dyad{1}_M + \frac{1}{n} \sum_{i=2}^{n} (\rho^{\ot i-1} \ot \sigma^{n-i})_{(1:n)_i} \ot \dyad{i}_M, \\
     %%%%%%%%%%%%%%%%%%%%%%%%%%%%%%%%%%%
    %=& \,\, \frac{1}{n}\left(\sigma^n \ot \dyad{1}_R + \widetilde{\sigma}^{\,1} \ot \rho \ot \sigma^{n-2} \ot \dyad{2}_R + \ldots + \widetilde{\sigma}^{\,n-1} \ot \rho^{\ot n-1} \ot \dyad{n}_R \right),
\end{align}
where for clarity we denoted $(1\!:\!n)_i := A_{i+1}B_{i+1}A_2B_2 \ldots A_iB_i A_1B_1 \ldots A_nB_n$.  Noting that $\Tr_{A_1B_1}(\sigma^n) = \sigma^{n-1}$, the reduced state of the catalyst ($CC' = A_2B_2\ldots A_{n}B_{n}M$) after this protocol reads:
\begin{align}
    \tr_{A_{1}B_{1}}\mathcal{T}_{\mathcal{E}}\left(\rho_{A_1B_1}\ot \omega_{CC'}\right) &= \sum_{i=1}^{n} \frac{1}{n} \rho^{\ot i-1}_{2:i} \ot \sigma^{n-i}_{i+1:n}\ot \dyad{i}_M = \omega_{CC'}.
\end{align}
After labelling the single-particle state $\widetilde{\sigma}^{\, i} := \tr_{/i}(\sigma^n)$, the reduced state of the main system ($A_1 B_1$) can be written as:
\begin{align}
    \label{eq:final_state}
    \widetilde{\sigma}_{A_1B_1}' = \tr_{CC'} \mathcal{T}_{\mathcal{E}}\left(\rho_{A_1B_1}\ot \omega_{CC'}\right) &= \frac{1}{n}\left(\widetilde{\sigma}^{\, 0} + \widetilde{\sigma}^{\, 1} + \ldots + \widetilde{\sigma}^{\, n}\right) = \frac{1}{n} \sum_{i=1}^n \tr_{/i} \mathcal{E}(\rho^{\ot n}).
\end{align}
In particular, we see that the main system ends up in a certain averaged state, while the reduced state of the catalyst remains unchanged. The fully entangled fraction of the main system therefore becomes:
\begin{align}
    \label{app:eq_fid_rhop}
    f(\rho'_{A_1B_1}) &= \sum_{i=1}^n \langle\phi^+ | \tr_{/i}\mathcal{E}(\rho^{}_{AB})|\phi^+\rangle.
\end{align}

Finally, we use the state state $\rho'_{A_1B_1}$ processed via $\mathcal{T}_{\mathcal{E}}$ as an input to a standard noisy teleport scheme $\mathcal{T}_2$, defined in the main text. Optimising the protocol over all maps $\mathcal{E} \in \locc(A:B)$ leads to the average fidelity of teleportation:
\begin{align}
    \langle F \rangle_{\rho'} = \max_{\mathcal{E} \in \locc(A:B)}\frac{f(\rho'_{A_1B_1})d_R + 1}{d_R+1} =   \frac{n^{-1} f_n(\rho_{AB}^{\ot n})d_R + 1}{d_R+1}.
\end{align}
Taking the limit $n \rightarrow \infty$ in the above expression completes the proof.

\section{\chg{Correlations between the system and the catalyst in the catalytic teleportation protocol}}
Consider the state of the main system $(A_1B_1)$ and the catalyst $CC' \equiv A_2B_2 \ldots A_nB_nM$. After the action of the first part of the catalytic teleportation protocol, $\mathcal{T}_{\mathcal{E}}$, the joint state becomes:
\begin{align}
    \mathcal{T}_{\mathcal{E}}(\rho_{A_1B_1} \ot \omega_{CC'}) = \sum_{i=0}^{n-1} \frac{1}{n} ({\rho^{\ot i} \ot \sigma^{n-i}})_{(1:n)_i} \ot \dyad{i+1}_M.
\end{align}
In order to show that the correlations between the system and the catalyst can be made as small as possible we now make an additional assumption. In particular, we further assume that the collective transformation ${\mathcal{E}}$ produces, approximately, a tensor-product state. This leads to the following lemma:
\begin{lemma}
Let $\mathcal{T}_{\mathcal{E}}$ be the protocol defined in Appendix $A$. If:
\begin{align}
    \label{eq:lemm_assump}
    \norm{\mathcal{E}(\rho_{AB}^{\ot n}) - \sigma^{\ot n}_{AB}}_1 \leq \epsilon,
\end{align}
then
\begin{align}
    \norm{\mathcal{T}_{\mathcal{E}}(\rho_{A_1B_1} \ot \omega_{CC'}) - \widetilde{\sigma}_{A_1B_1}' \ot \omega_{CC'}}_1 \leq \mathcal{O}(\epsilon),
\end{align}
where $\widetilde{\sigma}_{A_1B_1}'$ is defined in Eq. (\ref{eq:final_state}).
\end{lemma}
\begin{proof}
Let $\mathcal{U}_i$ be a unitary map that swaps subsystem $A_1B_1$ with $A_iB_i$, in particular
\begin{align}
    \mathcal{U}_i(\rho^{\ot i-1} \ot \sigma^{n-i+1}) = (\rho^{\ot i-1} \ot \sigma^{n-i+1})_{(1:n)_i}. 
\end{align}
This allows us to write
\begin{align}
    \norm{\mathcal{T}_{\mathcal{E}}(\rho_{A_1B_1} \ot \omega_{CC'}) - \widetilde{\sigma}_{A_1B_1}' \ot \omega_{CC'}}_1 &= \frac{1}{n} \norm{\sum_{i=1}^n  \mathcal{U}_i(\rho^{\ot i-1} \ot \sigma^{n-i+1}) \ot \dyad{i} - \sum_{i=1}^n \widetilde{\sigma} \ot \rho^{\ot i-1} \ot \sigma^{n-i}\ot \dyad{i}}_1\\
    &\leq \frac{1}{n} \sum_{i=1}^n \norm{\mathcal{U}_i(\rho^{\ot i-1} \ot \sigma^{n-i+1}) - \widetilde{\sigma} \ot \rho^{\ot i-1} \ot \sigma^{n-i}}_1 \\
    &= \frac{1}{n} \sum_{i=1}^n \norm{\rho^{\ot i-1} \ot \sigma^{n-i+1} - \mathcal{U}_i^{\dagger}({\widetilde{\sigma} \ot \rho^{\ot i-1} \ot \sigma^{n-i}})}_1\\
    &= \frac{1}{n} \sum_{i=1}^n \norm{\rho^{\ot i-1} \ot \sigma^{n-i+1} - \rho^{\ot i-1} \ot \widetilde{\sigma} \ot \sigma^{n-i}}_1 \\
    &= \frac{1}{n} \sum_{i=1}^n \norm{\sigma^{n-i+1} - \widetilde{\sigma} \ot \sigma^{n-i}}_1,
\end{align}
where in the third line we have used the fact that trace distance is invariant under unitary transformations. Now, using the assumption from Eq. (\ref{eq:lemm_assump}) and the fact that $\sigma^n := \mathcal{E}(\rho^{\ot n})$, we can further write:
\begin{align}
    \norm{\sigma^{n-i+1} - \widetilde{\sigma} \ot \sigma^{n-i}}_1 &= \norm{\sigma^{n-i+1} - \sigma^{\ot n-i+1} + \sigma^{\ot n-i+1} - \widetilde{\sigma} \ot \sigma^{n-i}}_1 \\
    &\leq \norm{\sigma^{n-i+1} - \sigma^{\ot n-i+1}}_1 + \norm{\sigma \ot \sigma^{\ot n-i} - \widetilde{\sigma} \ot \sigma^{n-i}}_1 \\
    &= \norm{\sigma^{n-i+1} - \sigma^{\ot n-i+1}}_1 + \norm{\sigma - \widetilde{\sigma}}_1 + \norm{\sigma^{\ot n-i} - \sigma^{n-i}}_1 \\
    &\leq 2 \norm{\sigma^n - \sigma^{\ot n}} + \norm{\sigma - \widetilde{\sigma}}_1,
\end{align}
where in the second and third lines we used triangle inequality and the fact that trace distance is contractive under CPTP maps (partial trace). Let us now define a CPTP map $\Lambda$ acting as $\Lambda(\rho) = \frac{1}{n}\sum_{i=1}^n \tr_{/i}\rho$. Therefore, $\sigma = \Lambda(\sigma^{\ot n})$ and $\widetilde{\sigma} = \Lambda(\sigma^n)$. Using again the fact that trace distance is contractive under CPTP maps we can further write: 
\begin{align}
    2 \norm{\sigma^n - \sigma^{\ot n}}_1 + \norm{\sigma - \widetilde{\sigma}}_1 &= 2 \norm{\sigma^n - \sigma^{\ot n}}_1 + \norm{\Lambda(\sigma^{\ot n}) - \Lambda(\sigma^n)}_1   \\
    &\leq  3 \norm{\sigma^n - \sigma^{\ot n}}_1 \\
    & \leq 3 \epsilon = \mathcal{O}(\epsilon).
\end{align}
This proves the claim.
\end{proof}
Therefore, the amount of correlations built between the system and the catalyst, as measured by the trace distance, can be made arbitrarily small by increasing the size of the catalyst (as specified by $n$).

\section{Proof of Lemma 1}
Let us begin the proof by recalling the following well-known fact \cite{PhysRevA.63.012307}:
\begin{lemma}
\label{lem:app1}
For any two pure states $\psi_{AB}$ and $\psi'_{AB}$ there exists $\mathcal{E} \in \locc(A:B)$ such that for all $\epsilon > 0$ and sufficiently large $n$:
\begin{align}
    \mathcal{E}(\psi_{AB}^{\ot n}) = {\widehat{\psi}^{n}}_{AB} \quad \text{s.t.} \quad  \norm{{\widehat{\psi}^{n}}_{AB} - (\psi'_{AB})^{\ot n}}_1 \leq \epsilon,
\end{align}
if and only if:
\begin{align}
    S(\rho_A) \geq S(\rho_A'),
\end{align}
where $\rho_{A} = \tr_B \dyad{\psi_{AB}}$ and $\rho_{A}' = \tr_B \dyad{\psi_{AB}'}$.
\end{lemma}
Let $n$ be the smallest possible number of copies such that there exists a transformation $\mathcal{E} = \mathcal{E}^*$ from Lemma \ref{lem:app1}. Using this as our educated guess for the optimisation in the definition of $f_{\text{reg}}(\rho_{AB})$ we have that for all sufficiently large $n$:
\begin{align}
    f_{n}(\psi^{\ot n}_{AB}) &\geq \sum_{i=1}^{n} \langle \phi^+|\tr_{/i} \mathcal{E}^*(\psi_{AB}^{\ot n})|\phi^+\rangle \\
    &= \sum_{i=1}^{n} \langle \phi^+|\tr_{/i} \widehat{\psi}_{AB}^{n}|\phi^+\rangle \\
    &\geq\sum_{i=1}^{n} \left(\langle \phi^+|\psi_{AB}'|\phi^+\rangle - \epsilon \right) \\
    &= n f(\psi_{AB}') - n \epsilon.
\end{align}
Using this and the fact that $\epsilon$ can be made arbitrarily small we can infer that the regularised entanglement fraction can be lower-bounded by:
\begin{align}
    f_{\text{reg}}(\psi_{AB}) \geq f(\psi_{AB}')
\end{align}
for all $ \psi_{AB}'$ such that $S(\rho_{A}) \geq S(\rho_{A}')$. 

\section{\chgR{Catalytic advantage for small catalysts}}
In this section we give an explicit example demonstrating that catalytic advantages in teleportation can be demonstrated for low-dimensional states. More specifically, let us consider teleporting qutrit states ($d_R = 3$) using a pair of qutrits  $A_1B_1$ with $\text{dim}(A_1) = \text{dim}(B_1) = 3$, prepared in a maximally-entangled state on a two-level support. In addition, Alice and Bob also share a pair of qutrits, $A_2B_2$, with $\text{dim}(A_2) = \text{dim}(B_2) = 3$, and a classical register $M$ with $\text{dim}(M) = 2$, both of which acting as a catalyst. More specifically, consider $0 \leq x \leq 1$ and take
\begin{align}
    \ket{\psi}_{A_1B_1} = \frac{1}{\sqrt{2}} \left(\ket{00} + \ket{11} \right)_{A_1B_1}, \qquad \omega_{A_2B_2M} = \frac{1}{2} {\gamma}_{A_2B_2}  \ot \dyad{1}_M + \frac{1}{2} \dyad{\psi}_{A_2B_2} \ot \dyad{2}_M,
\end{align}
where $\gamma_{A_2B_2} = x \dyad{\phi_+}_{A_2B_2} + (1-x) \dyad{00}_{A_2B_2}$, and $\ket{\phi_+}_{A_2B_2} = \frac{1}{\sqrt{3}} \left(\ket{00} + \ket{11} + \ket{22}\right)_{A_2B_2}$. The key observation is that we can find a (relatively simple) map $\mathcal{E} \in \locc(A_1A_2:B_1B_2)$ that transforms $\ket{\psi}_{A_1B_1}\ket{\psi}_{A_2B_2}$ into a pure state
\begin{align}
    \label{eq:psi_tilde_def}
    \ket{\widetilde{\varphi}}_{A_1B_1A_2B_2} = \sqrt{x} \ket{00}_{A_1B_1} \ket{\phi_+}_{A_2B_2} + \sqrt{1-x} \ket{11}_{A_1B_1} \ket{00}_{A_2B_2}.
\end{align}
This can be checked verified by calculating the vectors of squared Schmidt coefficients, $\lambda$. More specifically, we have
\begin{align}
    \lambda(\ket{\psi}^{\ot 2}) &= \left[\frac{1}{4}, \frac{1}{4}, \frac{1}{4}, \frac{1}{4}, 0, 0, 0, 0, 0\right], \\
    \lambda(\ket{\widetilde{\varphi}}) &= \left[\frac{x}{3}, \frac{x}{3}, \frac{x}{3}, 1-x, 0, 0, 0, 0, 0 \right].
\end{align}
Nielsen's theorem \cite{Nielsen_1998} states that $\ket{\psi}^{\ot 2}$ can be transformed into $\ket{\widetilde{\varphi}}$ using a one-way LOCC if and only if $\lambda(\ket{\psi}^{\ot 2})$ is majorized by $\lambda(\ket{\widetilde{\varphi}})$, which is the case in our construction for all $x \geq 3/4$. Moreover, without loss of generality, $\mathcal{E}$ takes the form
\begin{align}
    \mathcal{E}\left( \rho \right) = \sum_i \left(E_i \ot U_i \right) \rho \left(E_i \ot U_i \right)^{\dagger},
\end{align}
for $i = 1, \ldots, d_A^2 = 9$ measurement operators $E_i$ satisfying $\sum_i E_i^{\dagger} E_i = \mathbb{1}$ and a collection of unitaries $\{U_i\}_i^{d_A^2}$. Therefore the protocol $\mathcal{E}$ requires at most $\log_2 9$ bits of communication from Alice to Bob, corresponding to communicating which of the measurement outcomes Alice has obtained. The explicit forms of these measurements and unitaries can be found using e.g. the methods described in Sec. $4.3$ in \cite{Nielsen2001}. 

Returning to Eq. (\ref{eq:psi_tilde_def}) we can verify that $\tr_{A_1B_1} \dyad{\widetilde{\varphi}}_{A_1B_1A_2B_2} = \gamma_{A_2B_2}$. The explicit subroutine for increasing entanglement fraction of $\ket{\psi}_{A_1B_1}$ is as follows. 

We start with a state of the form
\begin{align}
    \label{eq:app_small_dim_joint}
    \dyad{\psi}_{A_1B_1} \ot \omega_{A_2B_2 M} = \frac{1}{2} \left(\dyad{\psi}_{A_1B_1} \ot \gamma_{A_2B_2} \ot \dyad{1}_M + \dyad{\psi}_{A_1B_1} \ot \dyad{\psi}_{A_2B_2} \ot \dyad{2}_M \right). 
\end{align}
Alice measures the classical register $M$. If the register is in state $\dyad{1}_M$, she communicates the outcome to Bob and the parties do nothing and proceed to the next iteration of the protocol. When the register is in state $\dyad{2}_M$, Alice and Bob apply the map $\mathcal{E}$ to $A_1B_1A_2B_2$. This results in transforming the joint state from Eq. (\ref{eq:app_small_dim_joint}) into
\begin{align}
    \label{eq:app_small_dim_joint2}
    \frac{1}{2} \left(\dyad{\psi}_{A_1B_1} \ot \gamma_{A_2B_2} \ot \dyad{1}_M + 
    \dyad{\widetilde{\varphi}}_{A_1B_1A_2B_2} \ot \dyad{2}_M \right).
\end{align}
Alice and Bob now relabel their systems according to the following recipe:
\begin{align}
    (\cdot)_{A_1B_1 A_2 B_2} \ot \dyad{1}_M & \rightarrow (\cdot)_{A_2 B_2 A_1B_1} \ot \dyad{2}_M, \\\
  (\cdot)_{A_1B_1 A_2 B_2} \ot \dyad{2}_M & \rightarrow (\cdot)_{A_1B_1 A_2 B_2} \ot \dyad{1}_M.
\end{align}
This is only done to make our bookkeeping simpler and does not involve any physical processing. After this procedure, the joint state from Eq. (\ref{eq:app_small_dim_joint2}) can be written as
\begin{align}
    \dyad{\psi'}_{A_1B_1A_2B_2} := \frac{1}{2} \left( \dyad{\widetilde{\varphi}}_{A_1B_1A_2B_2} \ot \dyad{1}_M + \gamma_{A_1B_1} \ot \dyad{\psi}_{A_2B_2} \ot \dyad{2}_M \right).
\end{align}
It can be readily verified that the marginals satisfy
\begin{align}
    \tr_{A_1B_1} \dyad{\psi'}_{A_1B_1A_2B_2} &= \frac{1}{2} \left(\gamma_{A_2B_2} \ot \dyad{1}_M + \dyad{\psi}_{A_2B_2} \ot \dyad{2}_{M} \right), \\
     \tr_{A_2B_2M} \dyad{\psi'}_{A_1B_1A_2B_2} &= \frac{1}{2}  \left( \gamma'_{A_1B_1} + \gamma_{A_1B_1} \right),
\end{align}
where 
\begin{align}
    \gamma'_{A_1B_1} &:= \tr_{A_2B_2} \dyad{\widetilde{\varphi}}_{A_1B_1A_2B_2} \\ 
    &= x \dyad{00}_{A_1B_1} + (1-x) \dyad{11}_{A_1B_1} + \frac{\sqrt{x(1-x)}} {\sqrt{3}} \dyad{00}{11} + \frac{\sqrt{x(1-x)}} {\sqrt{3}} \dyad{11}{00}.
\end{align}
After tracing out the catalyst, Alice and Bob obtain the state
\begin{align}
    \rho_{A_1B_1}' = \frac{1}{2}(\gamma_{A_1B_1} + \gamma'_{A_1B_1}).
\end{align}
It can be verified that the  entanglement fraction of $\rho_{A_1B_1}'$ is given by
\begin{align}
    f(\rho_{A_1B_1}') = \langle \phi_+|\rho_{A_1B_1}'|\phi_+\rangle = \frac{1}{3}\left(1+\sqrt{\frac{x(1-x)}{3}} + x\right).
\end{align}
Optimizing this expression over all $x \in [\frac{3}{4}, 1]$ gives $x^* = \frac{1}{2} + \frac{\sqrt{3}}{4} \approx \frac{14}{15}$. The resulting entanglement fidelity $f(\rho_{A_1B_1}^*) = \frac{1}{2} + \frac{\sqrt{3}}{9} \approx 0.692$. Note that this choice fixes the state of the catalyst $\omega_{A_2B_2R}$, via $\gamma_{A_2B_2}$.

The overall procedure therefore allows the entanglement fraction of the shared entangled state to be (strictly) increased. Interestingly, the resulting state is mixed, as opposed to the initial pure state $\ket{\psi}_{A_1B_1}$. Using this preprocessed state in the standard teleportation protocol leads to the fidelity of teleportation
\begin{align}
    \langle F \rangle_{\rho^*} = \frac{ f(\rho_{A_1B_1}^*)  d_R + 1}{d_R + 1} &= \frac{3}{4} + \frac{1}{4}\left(\frac{1}{\sqrt{3}} - \frac{1}{2}\right)\approx 0.77 \\
    &> \frac{3}{4} = \langle F \rangle_{\psi}.
\end{align}
Entanglement catalysis can therefore improve the quality of the arising teleportation channel even when the system and the catalyst are small-dimensional. 

\section{\chg{Proof of Theorem 2}}
Our goal is show that there exist a channel $\mathcal{D} \in \mathscr{C}$ and a quantum state $\omega$ such that the following holds
\begin{align}   
    \label{eq:app_thm2_proof1}
    \min_{\mathcal{D} \in \mathscr{C}} \,\, \tr \left[\mathcal{D}(\rho_S \ot \omega_C)(O_S \ot \mathbb{1}_C)\right] &= \min_{\mathcal{E} \in \mathscr{C}}\,\, \frac{1}{n} \tr[\mathcal{D}(\rho_S^{\otimes n}) O_S^{\otimes n}] \\
    \tr_{C}\left[ \mathcal{D}(\rho_S \ot \omega_C)\right]&= \omega.
\end{align}
Let $\mathcal{E} \in \mathscr{C}$ be an arbitrary channel from the class $\mathscr{C}$ acting on $n$ copies of $\rho$ and consider the following catalyst [see also Eq. $(10)$ from the main text], i.e.
\begin{align}
    \omega_C = \frac{1}{n}\sum_{i=1}^{n} {\rho^{\ot i} \ot \sigma^{n-i}}\ot \dyad{i},
\end{align}
with $\sigma^n := \mathcal{E}(\rho^{\ot n})$ and $\sigma^{n-i} := \tr_{1:i} \sigma^n$ defined as before, with $n$ being an arbitrary natural number. Let us for simplicity label the quantum systems comprising the catalyst with $C_i$ and $M$ denote the classical register, i.e. $C = C_2\ldots C_n M$ and let us also label $S \equiv C_1$. The channel $\mathcal{D}$ acting on the system and the catalyst can be specified similarly as in Appendix $A$ using the following set of steps:
\begin{itemize}
    \item[(a)] A multi-copy transformation $\mathcal{E} \in \mathscr{C}$, is applied conditioned on the classical register of the catalyst, i.e
    \begin{align}
        \id_{C_1\ldots C_n} (\cdot) \ot \sum_{i=1}^{n-1} \langle i |\cdot |i \rangle_{M}+ \mathcal{E}_{C_1 \ldots C_n}(\cdot) \ot \langle n | \cdot | n \rangle_{M}.
    \end{align}
    \item[(b)] Cyclic relabelling of the classical register in the catalyst.
    \begin{align}
        \ket{1}_{M} \rightarrow \ket{2}_{M}, \qquad \ket{2}_{M} \rightarrow \ket{3}_{M}, \qquad \ldots, \qquad \ket{n}_{M} \rightarrow \ket{1}_{M}
    \end{align}
    \item[(c)] Cyclic relabelling of $SC = C_1 \ldots C_nM$ (system + catalyst).
    \begin{align}
        C_1 \rightarrow C_2 \qquad C_2 \rightarrow C_3, \qquad \ldots, \qquad C_n \rightarrow \rm C_1.
    \end{align}
\end{itemize}
It can be easily verified that a channel $\mathcal{D}$ defined in this way is in the class $\mathscr{S}$ and leaves the state of the catalyst invariant, i.e.
\begin{align}
    \tr_{S}\left[ \mathcal{D}(\rho_S \ot \omega_C)\right] = \omega_C.
\end{align}
Moreover $\mathcal{D}$ induces the following effective transformation on the system:
\begin{align}
    \label{eq:subr}
    \tr_C [\mathcal{D}(\rho_S \ot \omega_C)] = \frac{1}{n} \sum_i^{n} \tr_{/i}[\mathcal{E}(\rho_S^{\otimes n})].
\end{align}
Using this in Eq. (\ref{eq:app_thm2_proof1}) leads to the following equality
\begin{align}
    \label{eq:app_thm2_proof2}
    \tr \left[\mathcal{D}(\rho_S \ot \omega_C)(O_S \ot \mathbb{1}_C)\right] = \tr\left[ O_S \cdot \tr_C \mathcal{D}(\rho_S \ot \omega_C) \right] = \frac{1}{n} \tr[\mathcal{E}(\rho_S^{\otimes n}) O_S^{\otimes n}].
\end{align}
where $O^{\otimes n} \equiv \sum_{i=1}^n \mathbb{1}_{/i} \otimes O_{i}$. Since our choice of $\mathcal{E}$ was arbitrary, we can optimise both sides of Eq. (\ref{eq:app_thm2_proof2}) over all channels $\mathcal{E} \in \mathscr{C}$. This proves the theorem. 

\section{\chg{Application of Theorem 2 (work extraction in quantum thermodynamics)}}

To demonstrate Theorem $2$ we now apply it to the problem of work extraction in quantum thermodynamics. 

More specifically, let $\mathscr{C}$ be the set of all unitary channels and let $\mathcal{E}(\cdot) = U(\cdot)U^{\dagger}$ for some unitary matrix $U$. Moreover, let the observable be the system's total energy, that is $O_S = H_S$, where $H_S$ is the Hamiltonian associated with $\rho_S$. The goal is to lower the energy of system $S$. Since the whole process is unitary, no entropy is produced, and we can associate the total energy change on the system with thermodynamic work. The natural quantity that captures this is the so-called \emph{ ergotropy}, defined as
\begin{align}
    \label{eq:erg}
    \mathcal{W}_{\mathscr{C}}(\rho) = \max_{U_S} \left[ \tr[H_S \rho_S]- \tr[H_S U_S \rho_S U_S^{\dagger}] \right].
\end{align}
Importantly, the above is often smaller than the free energy of $\rho_S$, meaning that quantum states may contain a form of ``locked'' energy that cannot be extracted in a reversible process (so-called passive states, see e.g. \cite{Pusz1978}). 

Let us now use Theorem $2$ to find the amount of work that can be extracted using the help of some catalyst. To do so we specify Eq. ($23$) from the main text to our current setting, i.e.
\begin{align}
    \label{eq:app_thm2_app1}
    \max_{U_{SC}} \left[\tr[H_S \rho_S] -  \tr[U_{SC}(\rho_S \ot \omega_C)U_{SC}^{\dagger}(H_S \ot \mathbb{1}_C)]\right] =  \max_{\widetilde{U}_S} \left[ \tr[H_S \rho_S]- \frac{1}{n} \tr[\widetilde{U}_S\rho_S^{\otimes n}\widetilde{U}_S^{\dagger} H_S^{\otimes n}]  \right],
\end{align}
where $\widetilde{U}_S$ is an arbitrary unitary acting on $n$ copies of system $S$. Moreover, using Theorem $2$ we can also infer that $\tr_{S}\left[U_{SC}(\rho_S \ot \omega_C)U^{\dagger}_{SC}\right] = \omega_C$.

Interestingly, as shown in \cite{Alicki_2013}, the second term on the r.h.s. of Eq. (\ref{eq:app_thm2_app1}) approaches the minimal energy for fixed entropy, that is 
\begin{align}
    \lim_{n \rightarrow \infty} \frac{1}{n} \tr[U_S\rho_S^{\otimes n}U_S^{\dagger} H_S^{\otimes n}] = \tr[H_S \tau_S]
\end{align}
where $\tau_S = e^{-\beta H_S}/\tr[e^{-\beta H_S}]$ is the Gibbs state and $\beta$ is chosen such that the von Neumann entropy of $\tau_S$ is equal to that of $\rho_S$. In other words, when using catalysts one can unitarily extract all free energy locked in a quantum state. In fact, our choice of unitary operations was arbitrary and a similar analysis can be performed for the general notion of passivity and ergotropy, i.e. in Eq. (\ref{eq:erg}) one can choose any reasonable set of operations $\mathscr{C}$, for example the notion of local passivity considered in Ref. \cite{PhysRevLett.123.190601}.

\end{document}